\documentclass[12pt,a4paper]{article}
\usepackage{amsmath,amssymb}
\usepackage{setspace}
\usepackage{amsthm}
\usepackage{ascmac}
\usepackage{graphicx}
\usepackage{here}
\usepackage{url}

\usepackage{natbib}
\bibpunct[:]{(}{)}{;}{a}{}{;}
\usepackage[usenames]{color}
\usepackage{authblk}

\allowdisplaybreaks[1]

\newtheorem{theorem}{Theorem}[section]
\newcommand{\argmin}{\mathop{\rm argmin}\limits}
\makeatletter
  
  \@addtoreset{equation}{section}
\makeatother

\newtheorem{prop}{Proposition}[section]

\title{Robust and Sparse Regression via $\gamma$-divergence}
\date{} 

\author[1]{Takayuki Kawashima} 
\author[1,2,3]{Hironori Fujisawa} 

\affil[1]{ Department of Statistical Science, The Graduate University for Advanced Studies, Tokyo  \authorcr  E-mail: \texttt{t-kawa@ism.ac.jp}}

\affil[2]{The Institute of Statistical Mathematics, Tokyo  \authorcr E-mail: \texttt{fujisawa@ism.ac.jp} }

\affil[3]{ Department of Mathematical Statistics, Nagoya University Graduate School of Medicine  }

\begin{document}
\maketitle

\begin{abstract}
In high-dimensional data, many sparse regression methods have been proposed. However, they may not be robust against outliers. 
Recently, the use of density power weight has been studied for robust parameter estimation and the corresponding divergences have been discussed. 
One of such divergences is the $\gamma$-divergence and the robust estimator using the $\gamma$-divergence is known for having a strong robustness. 
In this paper, we consider the robust and sparse regression based on $\gamma$-divergence. 
We extend the $\gamma$-divergence to the regression problem and show that it has a strong robustness under heavy contamination even when outliers are heterogeneous. 
The loss function is constructed by an empirical estimate of the $\gamma$-divergence with sparse regularization and the parameter estimate is defined as the minimizer of the loss function. 
To obtain the robust and sparse estimate, we propose an efficient update algorithm which has a monotone decreasing property of the loss function. 
Particularly, we discuss a linear regression problem with $L_1$ regularization in detail. 
In numerical experiments and real data analyses, we see that the proposed method outperforms past robust and sparse methods. 
\end{abstract}

\section{Introduction}
In high-dimensional data, sparse regression methods have been intensively studied. 
The Lasso \citep{Tibshirani94regressionshrinkage} is a typical sparse linear regression method with $L_1$ regularization, but is not robust against outliers. 
Recently, robust and sparse linear regression methods have been proposed. 
The robust least angle regression (RLARS)\citep{RePEc:bes:jnlasa:v:102:y:2007:m:december:p:1289-1299} is a robust version of LARS \citep{efron2004}, which replaces the sample correlation by a robust estimate of correlation in the update algorithm. 
The sparse least trimmed squares (sLTS)\citep{alfons2013} is a sparse version of well-known robust linear regression method LTS \citep{10.2307/2288718} based on the trimmed loss function with $L_1$ regularization. 

Recently, the robust parameter estimation using density power weight has been discussed by \citet{oro24027}, \citet{basu_book}, \citet{oro2122}, \citet{windham}, \citet{Fujisawa:2008:RPE:1434999.1435056},  \citet{RePEc:oup:biomet:v:102:y:2015:i:3:p:559-572.}, and so on. 
The density power weight gives a small weight to the terms related to outliers and then the parameter estimation becomes robust against outliers. 
Among them, the $\gamma$-divergence proposed by \citet{Fujisawa:2008:RPE:1434999.1435056} is known for having a strong robustness, which implies that the latent bias can be sufficiently small even under heavy contamination. 
In addition, to obtain the robust estimate, an efficient update algorithm was proposed with a monotone decreasing property of the loss function. 

In this paper, we propose the robust and sparse regression problem based on $\gamma$-divergence. 
First, we extend the $\gamma$-divergence to regression problem and show a strong robustness under heavy contamination even when outliers are  heterogeneous. 
Next, we consider a loss function based on the $\gamma$-divergence with sparse regularization and propose an update algorithm to obtain the robust and sparse estimate. 
\citet{Fujisawa:2008:RPE:1434999.1435056} used a Pythagorian relation on the $\gamma$-divergence, but it is not compatible with sparse regularization. 
Instead of this relation, we use the Majorization-Minimization algorithm \citep{hunter:mm}. 
%
%
%
This idea is deeply considered in a linear regression problem with $L_1$ regularization. 
Finally, in numerical experiments and real data analyses, we show that our method outperforms other robust and sparse methods.

\section{Regression based on $\gamma$-divergence}
The $\gamma$-divergence was defined for two probability density functions and its properties were investigated by \citet{Fujisawa:2008:RPE:1434999.1435056}.
In this section, the $\gamma$-divergence is extended to the regression problem, in other words, defined for two conditional probability density functions  and the corresponding robust properties are presented. 

\subsection{$\gamma$-divergence for regression}\label{def-gamma}
We suppose $g(x,y)$, $g(y|x)$ and $g(x)$ are the underlying probability density functions of $(x,y)$, $y$ given $x$ and $x$, respectively. 
Let $f(y|x)$ be another parametric conditional probability density function of $y$ given $x$. 
Let us define the $\gamma$-cross entropy for regression by
\begin{align}\label{gamma-div emp}
& d_\gamma (g(y|x),f(y|x);g(x))  \nonumber \\
 &   = -\frac{1}{\gamma} \log \int \left( \int g(y|x) f(y|x)^{\gamma} dy  \right) g(x) dx  + \frac{1}{1+\gamma} \log \int \left( \int f(y|x)^{1+\gamma} dy  \right) g(x) dx \nonumber \\
 &   =  -\frac{1}{\gamma} \log \int  \int  f(y|x)^{\gamma} g(x,y) dxdy   + \frac{1}{1+\gamma} \log \int \left( \int f(y|x)^{1+\gamma} dy  \right) g(x) dx. 
\end{align}
The $\gamma$-divergence for regression is defined by
\begin{gather}\label{gamma-div def}
D_{\gamma}(g(y|x),f(y|x);g(x))= -d_{\gamma}(g(y|x),g(y|x);g(x))+d_{\gamma}(g(y|x),f(y|x);g(x)).
\end{gather}
The $\gamma$-divergence for regression was first proposed by \citet{Fujisawa:2008:RPE:1434999.1435056} and many properties were already shown. 
However, we adopt the definition (\ref{gamma-div def}), which is slightly different from the past one, because (\ref{gamma-div def}) satisfies the Pythagorian relation approximately (see Sec \ref{property}).
\begin{theorem}\label{thm:gam-div}
We can show that
\begin{align*}
\mathrm{(i)}& \quad  D_{\gamma}(g(y|x),f(y|x);g(x)) \geq 0  ,\\
\mathrm{(ii)}& \quad D_{\gamma}(g(y|x),f(y|x);g(x))=0 \quad \Leftrightarrow \quad  g(y|x)=f(y|x)  \quad \rm{(a.e.)}, \\
\mathrm{(iii)}& \quad \lim_{\gamma \to 0 } D_{\gamma}(g(y|x),f(y|x);g(x)) = \int D_{KL}(g(y|x),f(y|x)) g(x) dx, 
\end{align*}
where $D_{KL}(g(y|x),f(y|x))=\int g(y|x) \log g(y|x) dy - \int g(y|x) \log f(y|x)dy$.
\end{theorem}
The proof is in the Appendix. 
In what follows, we refer to the regression based on $\gamma$-divergence as the $\gamma$-regression.
%
%
\subsection{Estimation for $\gamma$-regression}
Let $f(y|x;\theta)$ be the conditional probability density function of $y$  given $x$ with parameter $\theta$. 
The target parameter can be considered by 
\begin{align}
\theta^*_{\gamma} & = \argmin_{\theta} D_{\gamma}(g(y|x),f(y|x;\theta);g(x)) \nonumber \\
&= \argmin_{\theta} d_{\gamma}(g(y|x),f(y|x;\theta);g(x)).\label{def est}
\end{align}
When $g(y|x)=f(y|x;\theta^*)$, we have $\theta^*_{\gamma}=\theta^*$.

Let $(x_1,y_1) ,\ldots, (x_n,y_n)$ be the observations randomly drawn from the underlying distribution $g(x,y)$. 
Using the formula (\ref{gamma-div emp}), the $\gamma$-cross entropy for regression, $ d_{\gamma}(g(y|x),f(y|x;\theta);g(x))$, can be empirically estimated by 
\begin{align*}
&\bar{d}_\gamma(f(y|x;\theta)) \\
&= - \frac{1}{\gamma} \log \left\{ \frac{1}{n} \sum_{i=1}^n  f(y_i | x_i ;\theta)^{\gamma} \right\}  + \frac{1}{1+\gamma} \log \left\{ \frac{1}{n} \sum_{i=1}^n \int f(y | x_i ;\theta)^{1+\gamma} dy \right\}. 
\end{align*}
By virtue of (\ref{def est}), we define the $\gamma$-estimator by
\begin{align*}
\hat{ \theta}_{\gamma}=\argmin_{\theta} \bar{d}_\gamma(f(y|x;\theta)).
\end{align*}
In a similar way to in \citet{Fujisawa:2008:RPE:1434999.1435056}, we can show the consistency of $\hat{\theta}_{\gamma}$ to $\theta^*_{\gamma}$ under some conditions.
\subsection{Robust properties for $\gamma$-regression}\label{property}
In this subsection, the robust properties are presented from two viewpoints of latent bias and Pythagorian relation. 
The latent bias was discussed in \citet{Fujisawa:2008:RPE:1434999.1435056} and \citet{RePEc:oup:biomet:v:102:y:2015:i:3:p:559-572.}, which is described later. 
Using the results obtained there, the Pythagorian relation is shown in Theorems \ref{pytha_homogene} and \ref{hetero_pytha}.

Let $f^*(y|x)=f_{\theta^*}(y|x)=f(y|x;\theta^*)$ and $\delta(y|x)$ be the target conditional probability density function and the contamination conditional probability density function related to outliers, respectively. 
Let $\epsilon$ and $\epsilon(x)$ denote the outlier ratios which are  independent and dependent of $x$, respectively. 
Under homogeneous and heterogeneous contaminations, we suppose that the underlying conditional probability density function can be expressed as 
\begin{align*}
g(y|x) &= (1-\epsilon)f(y|x;\theta^*) + \epsilon \delta(y|x), \\
g(y|x) &= (1-\epsilon(x))f(y|x;\theta^*) + \epsilon(x) \delta(y|x).
\end{align*}
Let 
\begin{align*}
 {\nu}_{f,\gamma }(x)  = \left\{ \int \delta (y|x) f(y|x)^{ \gamma } dy  \right\}^{ \frac{1}{ \gamma} } \qquad (   \gamma > 0).
\end{align*}
and let 
\begin{align*}
{\nu}_{f, \gamma} = \left\{ \int \nu_{f,\gamma}(x)^{\gamma} g(x) dx  \right\}^{ \frac{1}{ \gamma }} .
\end{align*} 
Here we assume that
\begin{align*}
\nu_{f_{\theta^*},\gamma} \approx 0,
\end{align*}
which implies that $\nu_{f_{\theta^*},\gamma}(x) \approx 0$ for any $x$ (a.e.) and  illustrates that the contamination conditional probability density  function $\delta(y|x)$ lies on the tail of the target conditional probability density function $f(y|x;\theta^*)$. 
For example, if $\delta(y|x)$ is the dirac function at the outlier $y_{\dag}(x)$ given $x$, then we have $\nu_{f_{\theta^*},\gamma}(x) = f(y_{\dag}(x)|x;\theta^*)$, which should be sufficiently small because $y_{\dag}(x)$ is an outlier. 
In this subsection, we show that $\theta^*_{\gamma}-\theta^*$ is expected to be small even if $\epsilon$ or $\epsilon(x)$ is not small. 
To make the discussion easier, we prepare the monotone transformation of the $\gamma$-cross entropy for regression by
\begin{align*}
\tilde{d}_{\gamma}(g(y|x),f(y|x;\theta);g(x))& = - \exp  \left\{ -\gamma d_{\gamma}(g(y|x),f(y|x;\theta);g(x)) \right\} \\
&= - \frac{  \int \left(  \int g(y|x) f(y|x;\theta)^{\gamma} dy \right) g(x) dx    }{ \left\{ \int \left( \int f(y|x;\theta)^{1+\gamma} dy  \right) g(x) dx \right\}^{\frac{\gamma}{1+\gamma} } }.
\end{align*}
\subsubsection{Homogeneous contamination}\label{homogeneous}
%
%
Here we provide the following proposition, which was given in  \citet{RePEc:oup:biomet:v:102:y:2015:i:3:p:559-572.}.
\begin{prop}\label{homo propos}
%
%
\begin{align*}
&\tilde{d}_{\gamma}(g(y|x),f(y|x;\theta);g(x)) \\
& = (1-\epsilon) \tilde{d}_\gamma (f(y|x;\theta^*),f(y|x;\theta);g(x)) -  \frac{\epsilon \nu_{f_{\theta},\gamma}^{\gamma}}{\left\{ \int \left( \int f(y|x;\theta)^{1+\gamma} dy  \right) g(x) dx \right\}^{\frac{\gamma}{1+\gamma} }} .
\end{align*}
\end{prop}

Recall that $\theta_{\gamma}^*$ and $\theta^*$ are also the minimizers of $\tilde{d}_\gamma (g(y|x),f(y|x;\theta);g(x))$ and $\tilde{d}_\gamma (f(y|x;\theta^*),f(y|x;\theta);g(x))$, respectively. 
We can expect $\nu_{f_{\theta},\gamma} \approx 0$ from the assumption $\nu_{f_{\theta^*},\gamma} \approx 0$ if the tail behavior of $f(y|x;\theta)$ is close to that of $f(y|x;\theta^*)$. 
We see from Proposition \ref{homo propos} and the condition $\nu_{f_{\theta},\gamma} \approx 0$ that
\begin{align*}
\theta^*_{\gamma} & =  \argmin_{\theta} \tilde{d}_\gamma (g(y|x),f(y|x;\theta);g(x)) \\
&=  \argmin_{\theta} \left[ (1-\epsilon) \tilde{d}_\gamma (f(y|x;\theta^*),f(y|x;\theta);g(x)) - \frac{  \epsilon  \nu_{f_{\theta},\gamma}^{\gamma} }{\left\{ \int \left( \int f(y|x;\theta)^{1+\gamma} dy  \right) g(x) dx \right\}^{\frac{\gamma}{1+\gamma} }} \right] \\
& \approx \argmin_{\theta}  (1-\epsilon) \tilde{d}_\gamma (f(y|x;\theta^*),f(y|x;\theta);g(x)) \\
& = \theta^* .
\end{align*} 
Therefore, under homogeneous contamination, it can be expected that the latent bias $\theta^*_{\gamma} - \theta^*$ is small even if $\epsilon$ is not small. 
%
Moreover, we can show the following theorem, using Proposition \ref{homo propos}.

\begin{theorem}\label{pytha_homogene}
Let $\nu = \max\left\{\nu_{f_{\theta},\gamma}  , \nu_{f_{\theta^*},\gamma} \right\}$. Then, the Pythagorean relation among $g(y|x)$, $f(y|x;\theta^*)$, $f(y|x;\theta)$ approximately holds:
\begin{align*}
&D_{\gamma}(g(y|x),f(y|x;\theta);g(x)) - D_{\gamma}(g(y|x),f(y|x;\theta^*);g(x)) -D_{\gamma}(f(y|x;\theta^*),f(y|x;\theta);g(x)) \\
&= O( \nu^{\gamma}).
\end{align*}
\end{theorem}
%
The proof is in the Appendix. 
The Pythagorean relation implies that the minimization of the divergence from $f(y|x;\theta)$ to the underlying conditional probability density function $g(y|x)$ is approximately the same as that to the target conditional probability density function $f(y|x;\theta^*)$.
Therefore, under homogeneous contamination, we can see why our proposed method works well in terms of minimization of $\gamma$-divergence.

\subsubsection{Heterogeneous contamination}

Under heterogeneous contamination, we assume that the parametric conditional probability density function $f(y|x;\theta)$ is a location-scale family given by 
\begin{align*}
f(y|x;\theta) = \frac{1}{\sigma} s \left(\frac{y- q(x;\xi)}{\sigma} \right) ,
\end{align*}
where $s(y)$ is a probability density function, $\sigma$ is a scale parameter and $q(x;\xi)$ is a location function with a regression parameter $\xi$, e.g., $q(x;\xi)=\xi^Tx$. 
Then, we can obtain
\begin{align*}
\int f(y|x;\theta)^{1+\gamma} dy &= \int \frac{1}{\sigma^{1+\gamma}} s\left(\frac{y- q(x;\xi)}{\sigma} \right)^{1+\gamma} dy \\
& = \sigma^{-\gamma} \int s(z)^{1+\gamma} dz,
\end{align*}
%
%
That does not depend on the explanatory variable $x$. 
%
%
Here, we provide the following proposition, which was given in \citet{RePEc:oup:biomet:v:102:y:2015:i:3:p:559-572.}. 
\begin{prop}\label{heter propos}
%
\begin{align*}
&\tilde{d}_{\gamma}(g(y|x),f(y|x;\theta);g(x)) \\
& = \left(1-\int \epsilon(x) g(x) dx \right)^{\frac{\gamma}{1+\gamma}} \tilde{d}_{\gamma} (f(y|x;\theta^*),f(y|x;\theta); (1-\epsilon(x)) g(x))  \\
& \qquad \qquad \quad - \frac{\int \nu_{f_{\theta},\gamma}(x)^{\gamma} \epsilon(x)g(x) dx  }{\left\{  \sigma^{-\gamma} \int s(z)^{1+\gamma} dz \right\}^{\frac{\gamma}{1+\gamma} }  }    .
\end{align*}
\end{prop}

The second term $\frac{\int \nu_{f_{\theta},\gamma}(x)^{\gamma} \epsilon(x)g(x) dx  }{\left\{  \sigma^{-\gamma} \int s(z)^{1+\gamma} dz \right\}^{\frac{\gamma}{1+\gamma} }  } $ can be approximated to be $0$ from the condition $\nu_{f_{\theta},\gamma} \approx 0$ and $\epsilon(x) < 1 $ as follows: 
\begin{align}\label{heter assump}
\frac{\int \nu_{f_{\theta},\gamma}(x)^{\gamma} \epsilon(x)g(x) dx  }{\left\{  \sigma^{-\gamma} \int s(z)^{1+\gamma} dz \right\}^{\frac{\gamma}{1+\gamma} }  } & < \frac{\int \nu_{f_{\theta},\gamma}(x)^{\gamma} g(x) dx  }{ \left\{  \sigma^{-\gamma} \int s(z)^{1+\gamma} dz \right\}^{\frac{\gamma}{1+\gamma} }  } \nonumber \\
&=\frac{ \nu_{f_{\theta},\gamma }^{\gamma} }{ \left\{  \sigma^{-\gamma} \int s(z)^{1+\gamma} dz \right\}^{\frac{\gamma}{1+\gamma} }  } \nonumber \\
& \approx 0 
\end{align}
%
%
%
We see from Proposition \ref{heter propos} and (\ref{heter assump}) that
\begin{align*}
\theta_{\gamma}^* & =  \argmin_{\theta} \tilde{d}_\gamma (g(y|x),f(y|x;\theta);g(x)) \\ 
\\
\\
&=  \argmin_{\theta} \left[ \left(1-\int \epsilon(x) g(x) dx \right)^{\frac{\gamma}{1+\gamma}} \tilde{d}_{\gamma} (f(y|x;\theta^*),f(y|x;\theta); (1-\epsilon(x)) g(x)) \right. \\
& \left. \qquad \qquad  \qquad - \frac{\int \nu_{f_{\theta},\gamma}(x)^{\gamma} \epsilon(x)g(x) dx  }{\left\{  \sigma^{-\gamma} \int s(z)^{1+\gamma} dz \right\}^{\frac{\gamma}{1+\gamma} }  } \right]   \\
& \approx  \argmin_{\theta} \left(1-\int \epsilon(x) g(x) dx \right)^{\frac{\gamma}{1+\gamma}} \tilde{d}_{\gamma} (f(y|x;\theta^*),f(y|x;\theta); (1-\epsilon(x)) g(x))  \\
& = \theta^* .
\end{align*} 
Therefore, under heterogeneous contamination in a location-scale family,  it can be expected that the latent bias $\theta_{\gamma}^* - \theta^*$ is  small even if $\epsilon(x)$ is not small. 
%
Moreover, we can show the following theorem, using Proposition \ref{heter propos}.
\begin{theorem}\label{hetero_pytha}
Let $\nu = \max\left\{\nu_{f_{\theta},\gamma} , \nu_{f_{\theta^*},\gamma} \right\}$. Then, the following relation among $g(y|x)$, $f(y|x;\theta^*)$, $f(y|x;\theta)$ approximately holds:
\begin{align*}
&D_{\gamma}(g(y|x),f(y|x;\theta);g(x)) \\
& - D_{\gamma}(g(y|x),f(y|x;\theta^*);g(x)) -D_{\gamma}(f(y|x;\theta^*),f(y|x;\theta);(1-\epsilon(x))g(x)) \\
&= O( \nu^{\gamma}).
\end{align*}
\end{theorem}
%
The proof is in the Appendix. 
The above is slightly different from a conventional Pythagorean relation,  because the base measure changes from $g(x)$ to $(1-\epsilon(x))g(x)$ in part. 
But, it also implies that the minimization of the divergence from $f(y|x;\theta)$ to the underlying conditional probability density function $g(y|x)$ is approximately the same as that to the target conditional probability density function $f(y|x;\theta^*)$.
Therefore, under heterogeneous contamination in a location-scale family, we can see why our proposed method works well in terms of minimization of $\gamma$-divergence.

\section{Sparse $\gamma$-regression}

\subsection{Estimation for sparse $\gamma$-regression}
The empirical estimation of the $\gamma$-cross entropy with a penalty term can be given by 
\begin{align*}
L_{\gamma}(\theta;\lambda) = \bar{d}_{\gamma}(f(y|x;\theta)) + \lambda P(\theta),
\end{align*}
where $P(\theta)$ is a penalty for parameter $\theta$ and $\lambda$ is a tuning parameter for the penalty term. 
As an example of penalty term, we consider $L_1$ (Lasso, \citealt {Tibshirani94regressionshrinkage}), elasticnet \citep{Zou05regularizationand}, and so on. 
The sparse $\gamma$-estimator can be proposed by
\begin{align}
 \hat{ \theta}_{S} = \argmin_{\theta} L_{\gamma}(\theta;\lambda) .
\end{align}
To obtain the minimizer, we propose the iterative algorithm by Majorization-Minimization algorithm (MM algorithm) \citep{hunter:mm}.
\subsection{MM algorithm for sparse $\gamma$-regression}%
The MM algorithm is constructed as follows. 
Let $h(\eta)$ be the objective function. 
Let us prepare the majorization function $h_{MM}$ satisfying
\begin{align*}
h_{MM}(\eta^{(m)}|\eta^{(m)})& = h(\eta^{(m)}), \\
h_{MM}(\eta|\eta^{(m)}) &\geq h(\eta) \ \ \mbox{ for all }  \eta,
\end{align*}
where $\eta^{(m)}$ is the parameter of the $m$-th iterative step for $m=0,1,2,\ldots$.  
Let us consider the iterative algorithm by
\begin{align*}
\eta^{(m+1)} = \argmin_{\eta} h_{MM}(\eta|\eta^{(m)}).
\end{align*}
Then, we can show that the objective function $h(\eta)$ monotonically decreases at each step, because
\begin{align*}
h(\eta^{(m)}) = h_{MM}(\eta^{(m)}|\eta^{(m)}) & \geq h_{MM}(\eta^{(m+1)}|\eta^{(m)}) \geq h(\eta^{(m+1)}).
\end{align*}
Note that $\eta^{(m+1)}$ is not necessary to be minimizer of $h_{MM}(\eta|\eta^{(m)})$. 
We only need $h_{MM}(\eta^{(m)}|\eta^{(m)}) \geq h_{MM}(\eta^{(m+1)}|\eta^{(m)})$. 

We construct the majorization function for the sparse $\gamma$-regression by the following inequality: 
\begin{align}\label{mm_ineq}
\kappa(z^{T}\eta) & \leq  \sum_{i} \frac{z_i \eta^{(m)}_i }{ z^{T}\eta^{(m)}} \kappa \left[ \eta_i \frac{z^{T}\eta^{(m)} }{\eta^{(m)}_i} \right],
\end{align}
where $\kappa(u)$ is a convex function, $z=(z_1,\ldots,z_n)^T$, $\eta=(\eta_1,\ldots,\eta_n)^T$, $\eta^{(m)}=(\eta_1^{(m)},\ldots,\eta_n^{(m)})^T$ and $z_i$, $\eta_i$ and $\eta^{(m)}_i$ are positive. 
The inequality (\ref{mm_ineq}) holds from Jensen's inequality. 
Here, we take $z_i=\frac{1}{n}$, $\eta_i = f(y_i | x_i ;\theta)^{\gamma}$, $\eta_i^{(m)} = f(y_i | x_i ;\theta^{(m)})^{\gamma}$ and $\kappa(u) =- \log u$ in (\ref{mm_ineq}). 
We can propose the majorization function as follows:
\begin{align*}
& h(\theta) \\
&= L_{\gamma}(\theta ; \lambda) \\
&= - \frac{1}{\gamma} \log \left\{ \frac{1}{n} \sum_{i=1}^n  f(y_i | x_i ;\theta)^{\gamma} \right\}  + \frac{1}{1+\gamma} \log \left\{ \frac{1}{n} \sum_{i=1}^n \int f(y | x_i ;\theta)^{1+\gamma} dy \right\}    + \lambda P(\theta)  \\
 & \leq  - \frac{1}{\gamma} \sum_{i=1}^n \frac{ \frac{1}{n} f(y_i | x_i ;\theta^{(m)})^{\gamma}  }  { \frac{1}{n} \sum_{l=1}^n  f(y_l | x_l ;\theta^{(m)})^{\gamma}  } \log  \left\{  f(y_i | x_i ;\theta)^{\gamma}  \frac{\frac{1}{n} \sum_{l=1}^n  f(y_l | x_l ;\theta^{(m)})^{\gamma}}{ f(y_i | x_i ;\theta^{(m)})^{\gamma}}  \right\}  \\
 & \qquad  \qquad  \qquad + \frac{1}{1+\gamma} \log \left\{ \frac{1}{n} \sum_{i=1}^n \int f(y | x_i ;\theta)^{1+\gamma} dy \right\}  + \lambda P(\theta) \\
&= - \sum_{i=1}^n \alpha^{(m)}_i \log f(y_i | x_i ;\theta) +  \frac{1}{1+\gamma} \log \left\{ \frac{1}{n} \sum_{i=1}^n \int f(y | x_i ;\theta)^{1+\gamma} dy \right\} \\
& \qquad \qquad \qquad + \lambda P(\theta) +const \\
&=h_{MM}(\theta|\theta^{(m)}) + const ,
\end{align*}
where $\alpha^{(m)}_{i}= f(y_i | x_i ;\theta^{(m)})^{\gamma} \bigg \slash \sum_{l=1}^n  f(y_l | x_l ;\theta^{(m)})^{\gamma}$ and $const$ is a term which does not depend on the parameter $\theta$.

The first term on the original target function $h(\theta)$ is a mixture type of densities, which is not easy to be optimized, while the first term on $h_{MM}(\theta | \theta^{(m)})$ is a weighted log-likelihood, which is often easy to be optimized. 
\subsection{Sparse $\gamma$-linear regression}\label{update gam linear}
Let $f(y|x;\theta)$ be the conditional density with $\theta=(\beta_0, \beta, \sigma^2)$, given by 
\begin{gather*}
f(y|x;\theta)=\phi(y;\beta_0+x^T\beta,\sigma^2),
\end{gather*} 
where $\phi(y;\mu,\sigma^2)$ is the normal density with mean parameter $\mu$ and variance parameter $\sigma^2$. 
Suppose that $P(\theta)$ is the $L_1$ regularization $||\beta ||_1$.
Then, after a simple calculation, we have 
\begin{align}
& h_{MM}(\theta|\theta^{(m)})= \frac{1}{2(1+\gamma)}  \log \sigma^{2} + \frac{1}{2} \sum_{i=1}^n \alpha^{(m)}_i \frac{(y_i -\beta_0 - x_{i}^{T}\beta )^{2}}{\sigma^{2}} +\lambda ||\beta||_1.
\end{align}
This function is easy to be optimized by an update algorithm. 
For a fixed value of $\sigma^2$, the function $h_{MM}$ is almost the same as Lasso except for the weight, so that it can be updated using the coordinate decent algorithm with a decreasing property of the loss function. 
For a fixed value of $(\beta_0,\beta^T)^T$, the function $h_{MM}$ is easy  to be minimized. 
Consequently, the update formula can be obtained by
\begin{align}
\beta_0^{(m+1)} &=  \sum_{i=1}^n\alpha_i^{(m)} (y_i-{x_i}^T \beta^{(m)}) ,\\
\beta^{(m+1)}_{j} &=  S\left(\sum_{i=1}^n \alpha^{(m)}_{i}(y_i-\beta^{(m+1)}_0 -r^{[-j]}_{i})x_{ij} ,{\sigma^2}^{(m)}\lambda \right) \bigg/  \left( \sum_{i=1}^n \alpha^{(m)}_i x^{2}_{ij} \right)  \quad(j=1,\ldots,p), \\
{\sigma^2}^{(m+1)}&= (1+\gamma) \sum_{i=1}^n \alpha^{(m)}_i (y_i -\beta^{(m+1)}_0 - x^{T}_i \beta^{(m+1)})^{2},
\end{align}
where $S(t,\lambda)=\textrm{sign}(t)(|t|-\lambda)_+$ and $r^{[-j]}_{i}=\sum_{k \neq j} x_{ik}\beta_k$, and the decreasing property holds: 
\begin{align}\label{dec pro}
h_{MM} (\theta^{(m+1)} | \theta^{(m)} ) \leq h_{MM} (\theta^{(m)} | \theta^{(m)} ).
\end{align}
It should be noted that $h_{MM}$ is convex with respect to parameter $\beta_0$, $\beta$ and has the global minimum with respect to parameter $\sigma^2$, but not convex with respect to some of them, so that the starting value of update formula is important. 
This issue is often important in robust statistics and discussed in Sec \ref{opt}. 
%

%
%
%
%

Moreover, we explain how to implement coordinate decent algorithm. 
In practice, we also use active set strategy \citep{friedman2007pathwise}.
Specifically, for a given $\beta^{(m)}$, we only update the non-zero coordinates of $\beta^{(m)}$ in the active set, until they are converged. 
Then, the non-active set parameter estimates are updated once. 
When they remain zero, the coordinate descent algorithm stops.

\subsection{Robust cross-validation}
In sparse regression, a regularization parameter is often selected via some criterion. 
Cross-validation is often used for selecting the regularization parameter. 
However, the ordinal cross-validation will fail due to outliers.
Therefore, we propose the robust cross-validation based on the $\gamma$-cross entropy. 
Let $\hat{\theta}_{\gamma}$ be the robust estimate based on the $\gamma$-cross entropy. 
The ordinal cross validation is based on the squared error and it can also be constructed using the KL-cross entropy with normal density. 
The cross-validation based on the $\gamma$-cross entropy can be given by 
\begin{align*}
\mbox{RoCV}(\lambda) &= -\frac{1}{\gamma_0} \log \left\{ \frac{1}{n} \sum_{i=1}^n f(y_i|x_i ;\hat{\theta}_{\gamma}^{[-i]} )^{\gamma_0} \right\} \\ 
& \quad + \frac{1}{1+\gamma_0} \log \left\{ \frac{1}{n} \sum_{i=1}^n \int f(y|x_i; \hat{\theta}_{\gamma}^{[-i]} )^{1+\gamma_0} dy  \right\},
\end{align*}
where $\hat{\theta}_{\gamma}^{[-i]}$ is the $\gamma$-estimator deleting the $i$-th observation and $\gamma_0$ is an appropriate tuning parameter. 
We can also adopt the $K$-fold cross validation to reduce the computational task \citep{hastie_09_elements-of.statistical-learning}. 

Here, we give a small modification on the above. 
We often focus only on the mean structure for prediction, not on the variance parameter. 
Therefore, in this paper, $\hat{\theta}_{\gamma}^{[-i]}$ is replaced by $(\hat{\beta}_{\gamma}^{[-i]}, \hat{\sigma}_{\gamma}^2 )$.

\subsection{Redescending Property}

First, we review a redescending property on M-estimation (see, e.g.,  \citealt{maronna:martin:yohai:2006}), which is often used in robust statistics. 
Suppose that the estimating equation is given by $\sum_{i=1}^n \zeta(z_i;\theta)=0$. Let $\hat{\theta}$ be a solution of the estimating equation. 
The bias caused by outlier $z_o$ is expressed as $\hat{\theta}_{n=\infty} - \theta^*$, where $\hat{\theta}_{n=\infty}$ is the limiting value of $\hat{\theta}$ and $\theta^*$ is the true parameter. 
We hope the bias is small even if the outlier $z_o$ exists. 
Under some conditions, the bias can be approximated to $\epsilon {\rm IF}(z_o;\theta^*)$, where $\epsilon$ is a small outlier ratio and ${\rm IF}(z;\theta^*)$ is the influence function. 
The bias is expected to be small when the influence function is small. 
The influence function can be expressed as ${\rm IF}(z;\theta^*)=A\zeta(z;\theta^*)$, where $A$ is a matrix independent of $z$, so that the bias is also expected to be small when $\zeta(z_o;\theta^*)$ is small. 
In particular, the estimating equation is said to have a redescending property if $\zeta(z;\theta^*)$ goes to zero as $||z ||$ goes to infinity. 
This property is favorable in robust statistics, because the bias is expected to be sufficiently small when $z_o$ is very large. 

Here, we prove a redescending property on the sparse $\gamma$-linear regression, i.e., when $f(y|x;\theta) = \phi(y;\beta_0+x^T\beta,\sigma^2)$ with $\theta=(\beta_0,\beta,\sigma^2)$ for fixed $x$. 
%
%
%
Recall that the estimate of the sparse $\gamma$-linear regression is the minimizer of the loss function 
\begin{align*}
L_{\gamma}(\theta;\lambda)= -\frac{1}{\gamma} \log \left\{ \frac{1}{n}\sum_{i=1}^n \phi(y_i;\beta_0+{x_i}^T\beta,\sigma^2)^{\gamma} \right\} +b_{\gamma}(\theta;\lambda),
\end{align*}
where
\begin{align*}
b_{\gamma} (\theta;\lambda)= \frac{1}{1+\gamma} \log \left\{ \frac{1}{n}  \sum_{i=1}^n \int \phi(y;\beta_0+{x_i}^T\beta,\sigma^2)^{1+\gamma} dy \right\}+ \lambda || \beta ||_1  . 
\end{align*}
Then, the estimating equation is given by
\begin{align*}
0 &= \frac{\partial}{\partial\theta} L_\gamma(\theta;\lambda) \\
 &= - \frac{\sum_{i=1}^n \phi(y_i;\beta_0+{x_i}^T\beta,\sigma^2)^\gamma s (y_i|x_i;\theta)}{\sum_{i=1}^n \phi(y_i;\beta_0+{x_i}^T\beta,\sigma^2)^\gamma } + \frac{\partial}{\partial \theta} b_\gamma(\theta;\lambda),
\end{align*}
where
\begin{align*}
s(y|x;\theta) = \frac{\partial \log \phi(y;\beta_0+x^T\beta,\sigma^2)}{ \partial \theta}. 
\end{align*}
This can be expressed by the M-estimation formula given by
\begin{align}
0 = \sum_{i=1}^n \psi(y_i|x_i;\theta),
\label{eq:MEE}
\end{align}
where
\begin{align*}
\psi(y|x;\theta) = \phi(y;\beta_0+x^T\beta,\sigma^2)^\gamma s (y|x;\theta) - \phi(y;\beta_0+x^T\beta,\sigma^2)^\gamma \frac{\partial}{\partial \theta} b_\gamma(\theta;\lambda).
\end{align*}
We can easily show that as $||y||$ goes to infinity, $\phi(y;\beta_0+x^T\beta,\sigma^2)$ goes to zero and $\phi(y;\beta_0+x^T\beta,\sigma^2) s (y|x;\theta)$ also goes to zero. 
Therefore, the function $\psi(y|x;\theta)$ goes to zero as $||y||$ goes to infinity, so that the estimating equation has a redescending property.

\section{Numerical experiment}\label{experiment}
In this section, we compare our method (Sparse $\gamma$-linear regression) with the representative sparse linear regression method, {\it Least absolute shrinkage and selection operator} (Lasso)\citep{Tibshirani94regressionshrinkage}, and the robust and sparse regression methods, {\it sparse least trimmed squares} (sLTS)\citep{alfons2013} and {\it robust least angle regression} (RLARS)\citep{RePEc:bes:jnlasa:v:102:y:2007:m:december:p:1289-1299}.
\subsection{Regression models for simulation}
We used the simulation model given by 
\begin{gather*}
y=\beta_0+\beta_1 x_1 + \beta_2 x_2+ \cdots +\beta_p x_p + e, \quad e \sim N(0,0.5^2). 
\end{gather*}
The sample size and the number of explanatory variables were set to be $n=100$ and $p=100,200$, respectively. 
The true coefficients were given by 
\begin{gather}
\beta_1 =1 ,\ \beta_2 = 2,\ \beta_4 = 4,\ \beta_7 = 7,\ \beta_{11} =11,\nonumber \\
 \beta_j =0 \ \mbox{for} \ j \in \{0, \ldots ,p \} \backslash \{1,2,4,7,11\}. 
\end{gather}
We arranged a broad range of regression coefficients to observe sparsity for various degrees of regression coefficients. 
The explanatory variables were generated from a normal distribution $N(0,\Sigma)$, where the covariance matrix $\Sigma=(\sigma_{ij} )_{1 \leq i,j \leq p }$ was given by $\sigma_{ij}=\rho^{|i-j|}$. 
We generated 100 random samples. 

Outliers were incorporated into simulations. 
We investigated two cases for outlier ratio ($\epsilon=0.1 \mbox{ and } 0.3$) and two cases for outlier pattern: (a) The outliers were generated around the middle  part of the explanatory variable, where the explanatory variables were generated from $N(0,0.5^2)$ and the error terms were generated from $N(20,0.5^2)$. 
(b) The outliers were generated around the edge part of the explanatory variable, where the explanatory variables were generated from $N(-1.5,0.5^2)$ and the error terms were generated from $N(20,0.5^2)$. 

\subsection{Performance measure}
The root mean squared prediction error (RMSPE) and mean squared error (MSE) were examined to verify predictive performance and fitness of regression coefficient: 
\begin{align*}
\textrm{RMSPE}(\hat{\beta}) &= \sqrt{\frac{1}{n}\sum_{i=1}^n (y_i^{*}- {x_i^{*}}^T \hat{\beta} )^2 }, \\
\textrm{MSE} &= \frac{1}{p+1} \sum_{j=0}^p (\beta_j^* - \hat{\beta_j})^2, \\
\end{align*}
where $y_i^*$ and $x_i^*$($i=1, \ldots ,n$) is the test sample generated from the simulation model without outliers and $\beta_j^*$'s are the true coefficients. 
The true positive rate (TPR) and true negative rate (TNR) were also reported to verify the sparsity: 
\begin{align*}
\textrm{TPR}(\hat{\beta}) &= \frac{| \{ j \in \{1,\ldots,p\}: \hat{\beta_j} \neq 0 \wedge  \beta_j^* \neq 0 \} |}{|\{ j \in \{1,\ldots,p\}:\beta_j^* \neq 0 \}|},\\
\textrm{TNR}(\hat{\beta}) &= \frac{| \{ j \in \{1,\ldots,p\}: \hat{\beta_j} = 0 \wedge  \beta_j^* = 0 \} |}{|\{ j \in \{1,\ldots,p\}:\beta_j^* = 0 \}|}.
\end{align*}
\subsection{Comparative methods}\label{init}
In this subsection, we explain three comparative methods; Lasso, RLARS, sLTS.

The Lasso is performed by R-package "glmnet". 
The regularization parameter $\lambda_{lasso}$ is selected by grid search via cross-validation in "glmnet". 
We used "glmnet" by default. 

The RLARS is performed by R-package "robustHD". 
This is a robust version of LARS \citep{efron2004}. 
The optimal model is selected via BIC by default. 
%
%

The sLTS is performed by R-package "robustHD". 
The sLTS has the regularization parameter $\lambda_{sLTS}$ and the fraction parameter  $\alpha$ of squared residuals used for trimmed squares. 
First, the regularization parameter $\lambda_{sLTS}$ is selected by grid search via BIC. 
The number of grids is 40 by default. 
However, we considered that this would be small under heavy contamination. 
Therefore, we used 80 grids under heavy contamination. 
Next, the fraction parameter $\alpha$ is 0.75 by default. 
In the case of $\alpha=0.75$, the ratio of outlier is less than 25\%. 
We considered this would be small under heavy contamination and large under low contamination in terms of statistical efficiency. 
Therefore, we used 0.65, 0.75, 0.85 as $\alpha$ under low contamination and 0.50, 0.65, 0.75 under heavy contamination.

\subsection{Details of our method}\label{opt}
\subsubsection{Initial points}
In our method, we need an initial point to obtain the estimate, because we use the iterative algorithm proposed in Sec \ref{update gam linear}. 
The estimate of other conventional robust and sparse regression methods  would give a good initial point. 
For another choice, the estimate of the RANSAC (random sample consensus) algorithm would also give a good initial point. 
In this experiment, we used the estimate of the sLTS as an initial point. 
%


\subsubsection{How to choose tuning parameters}
In our method, we have to choose some tuning parameters. 
The parameter $\gamma$ in $\gamma$-divergence was set to $0.1$ or $0.5$. 
The parameter $\gamma_0$ in robust cross-validation was set to $0.5$. 
In our experience, the result via Rocv is not sensitive to the selection of $\gamma_{0}$ when $\gamma_{0}$ is enough large, e.g. $\gamma_{0} =0.5,1$. 
The parameter $\lambda$ of $L_1$ regularization is often selected via grid search. 
We used 50 grids on range $[0.05\lambda_0,\lambda_0]$ with log scale, where $\lambda_0$ is an estimate of $\lambda$ which would shrink regression coefficients to zero.

\subsection{Result}
Table 1 is the low contamination case with outlier pattern (a). 
For the RMSPE, our method outperformed other comparative methods (the oracle value of the RMSPE is 0.5). 
For the TPR and TNR, the sLTS showed a similar performance to our method. 
The Lasso presented the worst performance, because it is sensitive to outliers. 
Table 2 is the heavy contamination case with outlier pattern (a). 
For the RMSPE, our method outperformed other comparative methods except in the case ($p$, $\epsilon$, $\rho$) $=$ (100, 0.3, 0.2) for sLTS with $\alpha=0.5$. 
The Lasso also presented a worse performance and furthermore the sLTS with $\alpha=0.75$ showed the worst performance due to a lack of truncation. 
For the TPR and TNR, our method showed the best performance. 

\begin{table}[H]
\caption{The outlier pattern (a) with $p=100, \ 200$, $\epsilon=0.1$ and $\rho=0.2, \ 0.5$  }
 \scalebox{0.8}{
\hspace{-1.2cm} 
 \begin{tabular}{c | c c c c | c c c c }    
 \multicolumn{1}{c}{} & \multicolumn{4}{c}{$p=100$, $\epsilon=0.1$, $\rho=0.2$}& \multicolumn{4}{c}{$p=100$, $\epsilon=0.1$, $\rho=0.5$} \\
Methods  & RMSPE & MSE & TPR & TNR  & RMSPE & MSE & TPR & TNR \\ \hline
Lasso  & 3.04 & 9.72$\times10^{-2}$ & 0.936 & 0.909 & 3.1 & 1.05$\times10^{-1}$ & 0.952 & 0.918 \\
RLARS & 0.806 & 6.46$\times10^{-3}$ &  0.936 & 0.949 & 0.718  & 6.7$\times10^{-3}$ &  0.944 & 0.962 \\
sLTS($\alpha=0.85$,80grids)  & 0.626 & 1.34$\times10^{-3}$ & 1.0  & 0.964 & 0.599 & 1.05$\times10^{-3}$ &  1.0 & 0.966 \\
sLTS($\alpha=0.75$,80grids) & 0.651 & 1.71$\times10^{-3}$ &  1.0 & 0.961 & 0.623 & 1.33$\times10^{-3}$ &  1.0 & 0.961 \\
sLTS($\alpha=0.65$,80grids) & 0.685 & 2.31$\times10^{-3}$ & 1.0 & 0.957 & 0.668 & 1.76$\times10^{-3}$ & 1.0 & 0.961  \\
sparse $\gamma$-linear reg($\gamma$=0.1)&  0.557 & 6.71$\times10^{-4}$ &  1.0 & 0.966 &  0.561 & 6.99$\times10^{-4}$ &  1.0 & 0.965\\
sparse $\gamma$-linear reg($\gamma$=0.5) & 0.575 & 8.25$\times10^{-4}$ & 1.0 & 0.961 & 0.573 & 9.05$\times10^{-4}$ & 1.0 & 0.959 \\

  \multicolumn{1}{c}{} & \multicolumn{4}{c}{$p=200$, $\epsilon=0.1$, $\rho=0.2$}& \multicolumn{4}{c}{$p=200$, $\epsilon=0.1$, $\rho=0.5$} \\
Methods &   RMSPE & MSE & TPR & TNR &   RMSPE & MSE & TPR & TNR \\ \hline
Lasso   &  3.55 & 6.28$\times10^{-2}$ & 0.904 & 0.956 &  3.37 & 6.08$\times10^{-2}$ & 0.928 & 0.961 \\  
RLARS  &  0.88 & 3.8$\times10^{-3}$ & 0.904 & 0.977&  0.843 & 4.46$\times10^{-3}$ & 0.9 & 0.986\\ 
sLTS($\alpha=0.85$,80grids) & 0.631  & 7.48$\times10^{-4}$ & 1.0 & 0.972 &0.614  & 5.77$\times10^{-4}$ & 1.0 & 0.976  \\  
sLTS($\alpha=0.75$,80grids)  &  0.677 & 1.03$\times10^{-3}$ & 1.0 & 0.966 &  0.632 & 7.08$\times10^{-4}$ & 1.0 & 0.973 \\
sLTS($\alpha=0.65$,80grids) &   0.823 & 2.34$\times10^{-3}$ & 0.998 & 0.96&   0.7 & 1.25$\times10^{-3}$ & 1.0 & 0.967\\
sparse $\gamma$-linear reg($\gamma=0.1$)  &  0.58 & 4.19$\times10^{-4}$ & 1.0 & 0.981&  0.557 & 3.71$\times10^{-4}$ & 1.0 & 0.977 \\
sparse $\gamma$-linear reg($\gamma=0.5$) &  0.589 & 5.15$\times10^{-4}$  & 1.0 & 0.979 &  0.586 & 5.13$\times10^{-4}$  & 1.0 & 0.977  \\
\end{tabular}
}

\end{table}

\begin{table}[H]
\caption{The outlier pattern (a) with $p=100, \ 200$, $\epsilon=0.3$ and $\rho=0.2, \ 0.5$  }
 \scalebox{0.8}{
\hspace{-1.2cm} 
 \begin{tabular}{c | c c c c | c c c c }    
 \multicolumn{1}{c}{} & \multicolumn{4}{c}{$p=100$, $\epsilon=0.3$, $\rho=0.2$}& \multicolumn{4}{c}{$p=100$, $\epsilon=0.3$, $\rho=0.5$} \\
Methods  & RMSPE & MSE & TPR & TNR  & RMSPE & MSE & TPR & TNR \\ \hline
Lasso  & 8.07 & 6.72$\times10^{-1}$ & 0.806 & 0.903  &  8.1 & 3.32$\times10^{-1}$ & 0.8 & 0.952 \\
RLARS & 2.65  & 1.54$\times10^{-1}$ &  0.75 & 0.963  & 2.09  & 1.17$\times10^{-1}$ &  0.812 & 0.966 \\
sLTS($\alpha=0.75$,80grids)  & 10.4 & 2.08 &  0.886 & 0.709 &  11.7 & 2.36 &  0.854 & 0.67  \\
sLTS($\alpha=0.65$,80grids)  & 2.12 & 3.66$\times10^{-1}$ &  0.972 & 0.899 & 2.89 & 5.13$\times10^{-1}$ &  0.966 & 0.887 \\
sLTS($\alpha=0.5$,80grids) & 1.37 & 1.46$\times10^{-1}$ & 0.984 & 0.896 & 1.53 & 1.97$\times10^{-1}$ & 0.976 & 0.909  \\
sparse $\gamma$-linear reg($\gamma$=0.1)&  1.13 & 9.16$\times10^{-2}$ &  0.964 & 0.97 &  0.961 & 5.38$\times10^{-2}$ &  0.982 & 0.977 \\
sparse $\gamma$-linear reg($\gamma$=0.5) &  1.28 & 1.5$\times10^{-1}$ & 0.986 & 0.952  &1.00 & 8.48$\times10^{-2}$ & 0.988 & 0.958 \\

  \multicolumn{1}{c}{} & \multicolumn{4}{c}{$p=200$, $\epsilon=0.3$, $\rho=0.2$}& \multicolumn{4}{c}{$p=200$, $\epsilon=0.3$, $\rho=0.5$} \\
Methods &   RMSPE & MSE & TPR & TNR &   RMSPE & MSE & TPR & TNR \\ \hline
Lasso   &  8.11 & 3.4$\times10^{-1}$ & 0.77 & 0.951 & 8.02 & 6.51$\times10^{-1}$ & 0.81 & 0.91  \\  
RLARS  & 3.6 & 1.7$\times10^{-1}$ & 0.71 & 0.978 & 2.67 & 1.02$\times10^{-1}$ & 0.76 & 0.984 \\ 
sLTS($\alpha=0.75$,80grids) & 11.5  & 1.16 & 0.738 & 0.809 & 11.9  & 1.17& 0.78 & 0.811  \\  
sLTS($\alpha=0.65$,80grids)  &  3.34 & 3.01$\times10^{-1}$ & 0.94 & 0.929 & 4.22 & 4.08$\times10^{-1}$ & 0.928 & 0.924 \\
sLTS($\alpha=0.5$,80grids) &   4.02 & 3.33$\times10^{-1}$ & 0.892 & 0.903 &    4.94 & 4.44$\times10^{-1}$ & 0.842 & 0.909\\
sparse $\gamma$-linear reg($\gamma=0.1$)  &   2.03 & 1.45$\times10^{-1}$ & 0.964 & 0.924  &  3.2 & 2.86$\times10^{-1}$  & 0.94 & 0.936  \\
sparse $\gamma$-linear reg($\gamma=0.5$) & 1.23 & 7.69$\times10^{-2}$  & 0.988 & 0.942  &    3.13 & 2.98$\times10^{-1}$ & 0.944 & 0.94  \\
\end{tabular}
}

\end{table}

Table 3 is the low contamination case with outlier pattern (b). 
For the RMSPE, our method outperformed other comparative methods (the oracle value of the RMSPE is 0.5). 
For the TPR and TNR, the sLTS showed a similar performance to our method. 
The Lasso presented the worst performance, because it is sensitive to outliers. 
Table 4 is the heavy contamination case with outlier pattern (b). 
For the RMSPE, our method outperformed other comparative methods. 
The sLTS with $\alpha=0.5$ showed the worst performance. 
For the TPR and TNR, it seems that our method showed the best performance.

%
%
%
%
%
%
%
%
%
%
%
%

\begin{table}[H]
\caption{The outlier pattern (b) with $p=100, \ 200$, $\epsilon=0.1$ and $\rho=0.2, \ 0.5$  }
 \scalebox{0.8}{
\hspace{-1.2cm} 
 \begin{tabular}{c | c c c c | c c c c }    
 \multicolumn{1}{c}{} & \multicolumn{4}{c}{$p=100$, $\epsilon=0.1$, $\rho=0.2$}& \multicolumn{4}{c}{$p=100$, $\epsilon=0.1$, $\rho=0.5$} \\
Methods  & RMSPE & MSE & TPR & TNR  & RMSPE & MSE & TPR & TNR \\ \hline
Lasso  & 2.48 & 5.31$\times10^{-2}$ & 0.982 &0.518 & 2.84 & 5.91$\times10^{-2}$  & 0.98 & 0.565 \\
RLARS & 0.85  & 6.58$\times10^{-3}$ &  0.93 & 0.827 & 0.829 & 7.97$\times10^{-3}$  & 0.91 & 0.885\\
sLTS($\alpha=0.85$,80grids)  & 0.734 & 5.21$\times10^{-3}$ &  0.998 & 0.964 &  0.684 & 3.76$\times10^{-3}$ & 1.0 & 0.961 \\
sLTS($\alpha=0.75$,80grids) &0.66 & 1.78$\times10^{-3}$ &  1.0 & 0.975  &  0.648 & 1.59$\times10^{-3}$  & 1.0 & 0.961 \\
sLTS($\alpha=0.65$,80grids) & 0.734 & 2.9$\times10^{-3}$ & 1.0 & 0.96 &  0.66 & 1.74$\times10^{-3}$    &1.0  & 0.962  \\
sparse $\gamma$-linear reg($\gamma$=0.1)&  0.577 & 8.54$\times10^{-4}$ &  1.0 & 0.894 &   0.545  & 5.44$\times10^{-4}$  & 1.0 & 0.975\\
sparse $\gamma$-linear reg($\gamma$=0.5) & 0.581 &7.96$\times10^{-4}$ & 1.0 & 0.971 & 0.546 & 5.95$\times10^{-4}$  & 1.0 & 0.977 \\

  \multicolumn{1}{c}{} & \multicolumn{4}{c}{$p=200$, $\epsilon=0.1$, $\rho=0.2$}& \multicolumn{4}{c}{$p=200$, $\epsilon=0.1$, $\rho=0.5$} \\
Methods &   RMSPE & MSE & TPR & TNR &   RMSPE & MSE & TPR & TNR \\ \hline
Lasso   &  2.39 & 2.57$\times10^{-2}$ & 0.988 & 0.696  &  2.57 & 2.54$\times10^{-2}$  & 0.944 & 0.706 \\  
RLARS  &  1.01 & 5.44$\times10^{-3}$ & 0.896 & 0.923&  0.877 & 4.82$\times10^{-3}$  & 0.898 & 0.94\\ 
sLTS($\alpha=0.85$,80grids) & 0.708  & 1.91$\times10^{-3}$ & 1.0 & 0.975 &0.790  & 3.40$\times10^{-3}$   &0.994  &0.97  \\  
sLTS($\alpha=0.75$,80grids)  &  0.683 & 1.06$\times10^{-4}$ & 1.0 & 0.975 &   0.635& 7.40$\times10^{-4}$   & 1.0 & 0.977 \\
sLTS($\alpha=0.65$,80grids) &  1.11 & 1.13$\times10^{-2}$ & 0.984 & 0.956 &  0.768 & 2.60$\times10^{-3}$   &  0.998 & 0.968 \\
sparse $\gamma$-linear reg($\gamma=0.1$)  &  0.603 & 5.71$\times10^{-4}$ & 1.0 & 0.924&  0.563 & 3.78$\times10^{-3}$   & 1.0 & 0.979 \\
sparse $\gamma$-linear reg($\gamma=0.5$) &  0.592 & 5.04$\times10^{-4}$  & 1.0 & 0.982   &   0.566& 4.05$\times10^{-3}$    & 1.0 & 0.981  \\
\end{tabular}
}

\end{table}

\begin{table}[H]
\caption{The outlier pattern (b) with $p=100, \ 200$, $\epsilon=0.3$ and $\rho=0.2, \ 0.5$  }
 \scalebox{0.8}{
\hspace{-1.2cm} 
 \begin{tabular}{c | c c c c | c c c c }    
 \multicolumn{1}{c}{} & \multicolumn{4}{c}{$p=100$, $\epsilon=0.3$, $\rho=0.2$}& \multicolumn{4}{c}{$p=100$, $\epsilon=0.3$, $\rho=0.5$} \\
Methods  & RMSPE & MSE & TPR & TNR  & RMSPE & MSE & TPR & TNR \\ \hline
Lasso  & 2.81& 6.88$\times10^{-2}$   & 0.956 & 0.567   &  3.13 & 7.11$\times10^{-2}$   & 0.97 & 0.584 \\
RLARS & 2.70 & 7.69$\times10^{-2}$    & 0.872 & 0.789 &  2.22 & 6.1$\times10^{-2}$   & 0.852 & 0.855 \\
sLTS($\alpha=0.75$,80grids)  & 3.99   & 1.57$\times10^{-1}$   & 0.856 & 0.757 &   4.18& 1.54$\times10^{-1}$     & 0.878 & 0.771  \\
sLTS($\alpha=0.65$,80grids)  & 3.2 & 1.46$\times10^{-1}$  & 0.888 & 0.854 &   2.69 & 1.08$\times10^{-1}$   & 0.922& 0.867  \\
sLTS($\alpha=0.5$,80grids) &6.51 & 4.62$\times10^{-1}$  & 0.77 & 0.772 & 7.14& 5.11$\times10^{-1}$  & 0.844 & 0.778  \\
sparse $\gamma$-linear reg($\gamma$=0.1)& 1.75 & 3.89$\times10^{-2}$  & 0.974 & 0.725 &  1.47 & 2.66$\times10^{-2}$   & 0.976 & 0.865 \\
sparse $\gamma$-linear reg($\gamma$=0.5) &   1.68 & 3.44$\times10^{-2}$   & 0.98 & 0.782  & 1.65& 3.58$\times10^{-2}$  & 0.974 & 0.863 \\

  \multicolumn{1}{c}{} & \multicolumn{4}{c}{$p=200$, $\epsilon=0.3$, $\rho=0.2$}& \multicolumn{4}{c}{$p=200$, $\epsilon=0.3$, $\rho=0.5$} \\
Methods &   RMSPE & MSE & TPR & TNR &   RMSPE & MSE & TPR & TNR \\ \hline
Lasso   & 2.71 & 3.32$\times10^{-2}$   &  0.964 & 0.734 & 2.86 & 3.05$\times10^{-2}$ &  0.974 & 0.728 \\  
RLARS  & 3.03& 4.59$\times10^{-2}$  & 0.844 & 0.876 & 2.85 & 4.33$\times10^{-2}$   & 0.862 & 0.896\\ 
sLTS($\alpha=0.75$,80grids) &3.73  & 7.95$\times10^{-2}$    & 0.864 & 0.872&4.20 & 8.17$\times10^{-2}$     & 0.878 &0.87  \\  
sLTS($\alpha=0.65$,80grids)  & 4.45 & 1.23$\times10^{-1}$     & 0.85 & 0.886 & 3.61& 8.95$\times10^{-2}$  & 0.904 & 0.908 \\
sLTS($\alpha=0.5$,80grids) &  9.05 & 4.24$\times10^{-1}$   & 0.66 & 0.853 &     8.63& 3.73$\times10^{-1}$   & 0.748 & 0.864 \\
sparse $\gamma$-linear reg($\gamma=0.1$)  &   1.78& 1.62$\times10^{-2}$    & 0.994 & 0.731  &  1.82  & 1.62$\times10^{-2}$   & 0.988 & 0.844 \\
sparse $\gamma$-linear reg($\gamma=0.5$) & 1.79 & 1.69$\times10^{-2}$    & 0.988 & 0.79 &    1.77  &1.51$\times10^{-2}$   & 0.996 & 0.77  \\
\end{tabular}
}

\end{table}

%
%
%
%
%
%
%
%
%
%
%
\section{Real data analyses}
In this section, we use two real datasets to compare our method with comparative methods in real data analysis. 
We show the best result of comparative methods among some parameter situations (e.g., Sec \ref{init}), unless otherwise noted. 
\subsection{NCI-60 cancer cell panel}
We applied our method and comparative methods to regress protein expression on gene expression data at the cancer cell panel of the National Cancer Institute. 
Experimental conditions were set in the same way as in \citet{alfons2013} as follows. 
The gene expression data were obtained with an Affymetrix HG-U133A chip and normalized GCRMA method, resulting in a set of $p=22,283$ explanatory variables. 
The protein expressions based on 162 antibodies were acquired via reverse-phase protein lysate arrays and $\log_2$ transformed. 
One observation had to be removed since all values were missing in the gene expression data, reducing the number of observations to $n=59$. 
Then, the KRT18 antibody was selected as the response variable because it  had the largest MAD among 162 antibodies, i.e., KRT18 may include a large number of outliers. 
Both the protein expressions and the gene expression data can be downloaded via the web application CellMiner (\url{http://discover.nci.nih.gov/cellminer/}). 
As a measure of prediction performance, the root trimmed mean squared prediction error (RTMSPE) was computed via leave-one-out cross-validation given by 
\begin{align*}
\textrm{RTMSPE}=\sqrt{ \frac{1}{h} \sum_{i=1}^h (e)^2_{[i:n]} }, \ \ e^2=\left(\left(y_1 - x_1^{T} \hat{\beta}^{[-1]}  \right)^2 , \ldots , \left(y_n - x_n^{T} \hat{\beta}^{[-n]}  \right)^2 \right),
\end{align*}
where $(e)^2_{[1:n]} \leq   \cdots \leq (e)^2_{[n:n]}$ are the order statistics of $e^2$ and $h= \lfloor (n+1)0.75 \rfloor $. 
The choice of $h$ is important because it is preferable for estimating prediction performance that trimmed squares does not include outliers. 
We set $h$ in the same way as in \citet{alfons2013}, because the sLTS  detected 13 outliers in \citet{alfons2013}. 
In this experiment, we used the estimate of the RANSAC algorithm as an initial point instead of the sLTS because the sLTS required high computational cost in such a high dimensional data. 
\begin{table}[H]
\centering
\caption{ Root trimmed mean squared prediction error (RTMSPE) for protein expressions based on the KRT18 antibody (NCI-60 cancer cell panel data), computed from leave-one-out cross validation}
 \begin{tabular}{c c c}    
Methods &     RTMSPE & \# selected variables\\ \hline
Lasso   & 1.058 & 52 \\  
RLARS  & 0.936  & 18\\ 
sLTS & 0.721 & 33 \\  
sparse $\gamma$-linear reg($\gamma=0.1$) & 0.679 & 29\\
sparse $\gamma$-linear reg($\gamma=0.5$) & 0.700 & 30\\
\end{tabular}
\end{table}
%
%
%
Table 5 shows that our method outperformed other comparative methods
for the RTMSPE in high dimensional data. 
Our method presented the smallest RTMSPE with the second smallest number of explanatory variables. 
The RLARS presented the smallest number of explanatory
variables, but a much larger RTMSPE than our method. 
%
%

%

\subsection{Protein homology dataset}\label{protein_data}
We applied our method and comparative methods to protein sequences dataset used for KDD-Cup 2004. 
Experimental conditions were set in the same way as in  \citet{RePEc:bes:jnlasa:v:102:y:2007:m:december:p:1289-1299} as follows.
The whole dataset consists of $n=145,751$ protein sequences which has $153$ blocks corresponding to native protein. 
Each data point in particular block is a candidate homologous protein. 
There were 75 variables in the data set: the block number (categorical) and 74 measurements of protein features. 
The first protein feature was used as the response variable. 
Then, 5 blocks with a total of $n=4,141$ protein sequences were selected because they contained the highest proportions of homologous proteins (and hence the highest proportions of potential outliers). 
The data of each block was split into two almost equal parts to get a training sample of size $n_{tra} =2,072$ and a test sample of size $n_{test} = 2,069$. 
The number of explanatory variables was $p = 77$, consisting of 4 block indicators (variables 1-4) and 73 features. 
The whole protein, train and test dataset can be downloaded from  \url{http://users.ugent.be/~svaelst/software/RLARS.html}.
As a measure of prediction performance, the root trimmed mean squared prediction error (RTMSPE) was computed for the test sample given by 
\begin{align*}
\textrm{RTMSPE}= \sqrt{ \frac{1}{h} \sum_{i=1}^h (e)^2_{[i:n_{test}]} } , \ \ e^2=\left(\left(y_1 - {x_1}^{T} \hat{\beta}  \right)^2  , \ldots , \left(y_{n_{test}} - x_{n_{test}}^{T} \hat{\beta}  \right)^2 \right),
\end{align*}
where $(e)^2_{[1:n_{test}]} \leq   \cdots \leq (e)^2_{[n_{test}:n_{test}]}$ are the order statistics of $e^2$ and $h= \lfloor (n_{test}+1)0.99 \rfloor \mbox{ , } \lfloor (n_{test}+1)0.95 \rfloor \mbox{ or } \lfloor (n_{test}+1)0.9 \rfloor $.
In this experiment, we used the estimate of the sLTS as an initial point.
\begin{table}[H]
\centering
\caption{Root trimmed mean squared prediction error in the protein test set}
 \begin{tabular}{c c c c c }
\multicolumn{1}{c}{} & \multicolumn{3}{c}{trimming fraction } & \multicolumn{1}{c}{ } \\ \hline    
Methods & 1$\%$  & 5$\%$ & 10$\%$  & \# selected variables   \\ \hline
Lasso   & 10.697 & 9.66 & 8.729 & 22 \\  
RLARS  & 10.473 & 9.435 & 8.527 & 27\\ 
%
%
sLTS & 10.614 & 9.52 & 8.575  & 21\\  
%
%
%
sparse $\gamma$-linear reg($\gamma=0.1$) & 10.461 &  9.403 & 8.481 & 44\\
sparse $\gamma$-linear reg($\gamma=0.5$) & 10.463 &  9.369&  8.419 & 42\\
\end{tabular}
\end{table}
%
%
%
%
Table 6 shows that our method outperformed other comparative methods
for the RTMSPE. 
Our method presented the smallest RTMSPE with the largest number of explanatory variables. 
It might seem that other methods gave the smaller number of explanatory variables than necessary. 

\section*{Appendix}
\begin{proof}[Proof of Theorem \ref{def-gamma}]
For two non-negative functions $r(x,y)$ and $u(x,y)$ and probability density function $g(x)$, it follows from H\"{o}lder's inequality that
\begin{align*}
\int r(x,y)u(x,y)g(x)dxdy \leq \left( \int r(x,y)^{\alpha} g(x)dxdy \right)^{\frac{1}{\alpha}} \left( \int u(x,y)^{\beta} g(x)dxdy \right)^{\frac{1}{\beta}},
\end{align*}
where $\alpha$ and $\beta$ are positive constants and $\frac{1}{\alpha}+\frac{1}{\beta}=1$. 
The equality holds if and only if $r(x,y)^{\alpha}=\tau u(x,y)^{\beta}$ for a positive constant $\tau$. 
Let $r(x,y)= g(y|x)$, $u(x,y)=f(y|x)^{\gamma}$, $\alpha=1+\gamma$ and $\beta=\frac{1+\gamma}{\gamma}$. 
Then it holds that 
\begin{align*}
&\int \left( \int g(y|x)f(y|x)^{\gamma} dy \right) dg(x) \\
& \leq \left\{ \int \left( \int g(y|x)^{1+\gamma} dy \right) dg(x) \right\}^{\frac{1}{1+\gamma}} \left\{ \int \left( \int f(y|x)^{1+\gamma} dy \right) dg(x) \right\}^{\frac{\gamma}{1+\gamma}}. 
\end{align*} 
The equality holds if and only if $g(y|x)^{1+\gamma}=\tau (f(y|x)^{\gamma})^{\frac{1+\gamma}{\gamma}}$, i.e.  $g(y|x)= f(y|x)$ because $g(y|x)$ and $f(y|x)$ are conditional probability density functions.   
The properties (i), (ii) follow from this inequality, the equality condition and the definition of $D_{\gamma}(g(y|x),f(y|x);g(x))$. 

Let us prove the property (iii). 
Suppose that $\gamma$ is sufficiently small. 
Then, it holds that $f^{\gamma} = 1+ \gamma \log f + O(\gamma^2)$. 
The $\gamma$-divergence for regression is expressed by
\begin{align*}
&D_{\gamma}(g(y|x),f(y|x);g(x)) \\
& = \frac{1}{\gamma(1+\gamma)} \log \int \left\{ \int g(y|x)(1+\gamma\log g(y|x)+ O(\gamma^2)) dy \right\} g(x) dx \\
& \quad - \frac{1}{\gamma} \log \int \left\{ \int g(y|x) (1+\gamma \log f(y|x) + O(\gamma^2) ) dy \right\} g(x) dx \\
& \qquad \frac{1}{1+\gamma} \log \int \left\{ \int f(y|x)(1+\gamma \log f(y|x) +O(\gamma^2)) dy \right\} g(x)dx \\ 
& = \frac{1}{\gamma(1+\gamma)} \log \ \left\{ 1+ \gamma \int \left( \int g(y|x)\log g(y|x) dy \right) g(x) dx+ O(\gamma^2) \right\}  \\
& \quad - \frac{1}{\gamma} \log  \left\{ 1+\gamma \int \left( \int g(y|x) \log f(y|x)dy \right) g(x) dx+ O(\gamma^2)  \right\}  \\
& \qquad \frac{1}{1+\gamma} \log  \left\{  1+\gamma \int \left( \int f(y|x) \log f(y|x) dy \right) g(x) dx +O(\gamma^2) \right\}  \\ 
& = \frac{1}{(1+\gamma)} \int \left( \int g(y|x)\log g(y|x) dy \right) g(x) dx   \\
& \quad -  \int \left( \int g(y|x) \log f(y|x)dy \right) g(x) dx + O(\gamma)  \\
&= \int D_{KL}(g(y|x),f(y|x) ) g(x) dx +O(\gamma).  
\end{align*}
\end{proof}

\begin{proof}[Proof of Theorem \ref{pytha_homogene}]
We see that
\begin{align*}
& \int \left( \int g(y|x)f(y|x;\theta)^{\gamma} dy \right) g(x) dx \\
&= \int \left( \int \left\{  (1-\epsilon)f(y|x;\theta^*) + \epsilon \delta(y|x) \right\} f(y|x;\theta)^{\gamma} dy \right) g(x) dx \\
&=(1-\epsilon) \left\{ \int \left( \int  f(y|x;\theta^*)   f(y|x;\theta)^{\gamma} dy \right) g(x) dx \right\}  + \epsilon \left\{  \int \left( \int    \delta(y|x)  f(y|x;\theta)^{\gamma} dy \right) g(x) dx \right\}.
\end{align*}
It follows from the assumption $\epsilon < \frac{1}{2}$ that
\begin{align*}
\left\{ \epsilon \int \left( \int   \delta(y|x)  f(y|x;\theta)^{\gamma} dy \right)  g(x) dx \right\}^{\frac{1}{\gamma}} & < \left\{ \frac{1}{2} \int \left( \int    \delta(y|x) f(y|x;\theta)^{\gamma} dy \right) g(x) dx \right\}^{\frac{1}{\gamma }} \\
&< \left\{ \int \left( \int    \delta(y|x) f(y|x;\theta)^{\gamma} dy \right) g(x) dx \right\}^{\frac{1}{\gamma} }  =\nu_{f_{\theta},\gamma}   . 
\end{align*}
Hence,
\begin{align*}
\int & \left( \int g(y|x) f(y|x;\theta)^\gamma dy \right) g(x)dx \\
& =(1-\epsilon)\left\{ \int \left( \int f (y|x;\theta^*) f(y|x;\theta)^{\gamma} dy \right) g(x) dx \right\}  +O \left( \nu_{f_{\theta},\gamma}^\gamma \right) .
\end{align*}
Therefore, it holds that
\begin{align*}
& d_{\gamma}(g(y|x),f(y|x;\theta);g(x)) \\
&= -\frac{1}{\gamma} \log \int \left( \int g(y|x) f(y|x;\theta)^{\gamma} dy  \right) g(x) dx  + \frac{1}{1+\gamma} \log \int \left( \int f(y|x;\theta)^{1+\gamma} dy  \right) g(x) dx \\
&= -\frac{1}{\gamma} \log (1-\epsilon) -\frac{1}{\gamma} \log \int \left( \int f(y|x;\theta^*) f(y|x;\theta)^{\gamma} dy  \right) g(x) dx +O \left(  \nu_{f_{\theta},\gamma}^{\gamma}  \right) \\
 & \qquad \qquad \qquad \qquad + \frac{1}{1+\gamma} \log \int \left( \int f(y|x;\theta)^{1+\gamma} dy  \right) g(x) dx \\
&=  -\frac{1}{\gamma} \log (1-\epsilon) + d_\gamma (f(y|x;\theta^*),f(y|x;\theta);g(x))+O\left( \nu_{f_{\theta},\gamma}^{\gamma} \right) .
\end{align*}
Then, it follows that 
\begin{align*}
&D_{\gamma}(g(y|x),f(y|x;\theta);g(x)) - D_{\gamma}(g(y|x),f(y|x;\theta^*);g(x)) -D_{\gamma}(f(y|x;\theta^*),f(y|x;\theta);g(x)) \\[5pt]
&= \left\{ -d_{\gamma}(g(y|x),g(y|x);g(x)) +d_{\gamma}(g(y|x),f(y|x;\theta);g(x))   \right\} \\
& \quad - \left\{ -d_{\gamma}(g(y|x),g(y|x);g(x)) +d_{\gamma}(g(y|x),f(y|x;\theta^*);g(x))   \right\} \\
& \qquad -\left\{ -d_{\gamma}(f(y|x;\theta^*),f(y|x;\theta^*);g(x)) +d_{\gamma}(f(y|x;\theta^*),f(y|x;\theta);g(x))   \right\}   \\[5pt] 
&= d_{\gamma}(g(y|x),f(y|x;\theta);g(x)) - d_{\gamma}(f(y|x;\theta^*),f(y|x;\theta);g(x)) \\
& \qquad - d_{\gamma}(g(y|x),f(y|x;\theta^*);g(x)) +d_{\gamma}(f(y|x;\theta^*),f(y|x;\theta^*);g(x))   \\[5pt]
&= O\left( \nu^{\gamma} \right) .
\end{align*}

\end{proof}

\begin{proof}[Proof of Theorem \ref{hetero_pytha}]
We see that
\begin{align*}
& \int \left( \int g(y|x)f(y|x;\theta)^{\gamma} dy \right) g(x) dx \\
&= \int \left( \int \left\{  (1-\epsilon(x))f(y|x;\theta^*) + \epsilon(x) \delta(y|x) \right\} f(y|x;\theta)^{\gamma} dy \right) g(x) dx \\
&= \left\{ \int \left( \int  f(y|x;\theta^*)   f(y|x;\theta)^{\gamma} dy \right)(1-\epsilon(x)) g(x) dx  + \int \left( \int    \delta(y|x)  f(y|x;\theta)^{\gamma} dy \right) \epsilon(x) g(x) dx \right\}.
\end{align*}
It follows from the assumption $\epsilon(x) < \frac{1}{2} $ that
\begin{align*}
\left\{ \int \left( \int   \delta(y|x)  f(y|x;\theta)^{\gamma} dy \right) \epsilon(x) g(x) dx \right\}^{\frac{1}{\gamma}} & <  \left\{ \int \left( \int    \delta(y|x) f(y|x;\theta)^{\gamma} dy \right) \frac{ g(x)}{2} dx \right\}^{\frac{1}{\gamma}} \\
& <  \left\{ \int \left( \int    \delta(y|x) f(y|x;\theta)^{\gamma} dy \right) g(x) dx \right\}^{\frac{1}{\gamma}}  \\
& = \nu_{f_{\theta,\gamma}}.
\end{align*}
Hence,
\begin{align*}
\int & \left( \int g(y|x) f(y|x;\theta)^\gamma dy \right) g(x)dx \\
& =\left\{ \int \left( \int f (y|x;\theta^*) f(y|x;\theta)^{\gamma} dy \right) (1-\epsilon(x) ) g(x) dx \right\} +O( \nu_{f_\theta,\gamma}^\gamma) .
\end{align*}
Therefore, it holds that
\begin{align*}
& d_{\gamma}(g(y|x),f(y|x;\theta);g(x)) \\
&= -\frac{1}{\gamma} \log \int \left( \int g(y|x) f(y|x;\theta)^{\gamma} dy  \right) g(x) dx  + \frac{1}{1+\gamma} \log \int \left( \int f(y|x;\theta)^{1+\gamma} dy  \right) g(x) dx \\
&=  -\frac{1}{\gamma} \log \left\{ \int \left( \int f (y|x;\theta^*) f(y|x;\theta)^{\gamma} dy \right) (1-\epsilon(x) ) g(x) dx \right\} +O( \nu_{f_{\theta},\gamma}^\gamma) \\
 &   \qquad  + \frac{1}{1+\gamma} \log \int \left( \int f(y|x;\theta)^{1+\gamma} dy  \right) g(x) dx \\
&=  d_\gamma (f(y|x;\theta^*),f(y|x;\theta); (1-\epsilon(x)) g(x))  +O( \nu_{f_\theta,\gamma}^{\gamma}) \\
& \quad - \frac{1}{1+\gamma} \log \int \left( \int f(y|x;\theta)^{1+\gamma} dy  \right) (1-\epsilon(x))g(x) dx \\
& \qquad \quad  + \frac{1}{1+\gamma} \log \int \left( \int f(y|x;\theta)^{1+\gamma} dy  \right) g(x) dx \\
&=  d_\gamma (f(y|x;\theta^*),f(y|x;\theta); (1-\epsilon(x)) g(x))  +O( \nu_{f_\theta,\gamma}^{\gamma}) - \frac{1}{1+\gamma} \log \left\{1-\int \epsilon(x)g(x) dx \right\} .
\end{align*}
Then, it follows that 
\begin{align*}
&D_{\gamma}(g(y|x),f(y|x;\theta);g(x)) \\
&  - D_{\gamma}(g(y|x),f(y|x;\theta^*);g(x)) -D_{\gamma}(f(y|x;\theta^*),f(y|x;\theta);(1-\epsilon(x))g(x)) \\[5pt]
&= \left\{ -d_{\gamma}(g(y|x),g(y|x);g(x)) +d_{\gamma}(g(y|x),f(y|x;\theta);g(x))   \right\} \\
& \quad - \left\{ -d_{\gamma}(g(y|x),g(y|x);g(x)) +d_{\gamma}(g(y|x),f(y|x;\theta^*);g(x))   \right\} \\
& \quad -\left\{ -d_{\gamma}(f(y|x;\theta^*),f(y|x;\theta^*);(1-\epsilon(x))g(x)) +d_{\gamma}(f(y|x;\theta^*),f(y|x;\theta);(1-\epsilon(x))g(x))   \right\}   \\[5pt] 
&= d_{\gamma}(g(y|x),f(y|x;\theta);g(x)) - d_{\gamma}(f(y|x;\theta^*),f(y|x;\theta);(1-\epsilon(x))g(x)) \\
& \quad - d_{\gamma}(g(y|x),f(y|x;\theta^*);g(x)) +d_{\gamma}(f(y|x;\theta^*),f(y|x;\theta^*);(1-\epsilon(x))g(x))   \\[5pt]
%
%
&=  O\left( \nu^{\gamma} \right).
\end{align*}
\end{proof}

\bibliography{ref}

\end{document}